\documentclass{article} 

\usepackage{iclr2027_conference,times}

\usepackage[subtle]{savetrees}
\usepackage[table]{xcolor}
\usepackage{graphicx}
\usepackage{booktabs,microtype}
\usepackage{amsmath,amsfonts,amssymb,amsthm}
\usepackage{hyperref,url}
\usepackage{algorithmic, algorithm}

\hypersetup{pdftitle={Channel-Preserving Representation Alignment for EEG-to-Music Reconstruction},pdfauthor={Jiaxin Qing, Junwei Lu, Lexin Li},colorlinks=true,linkcolor=blue,citecolor=blue,urlcolor=blue}

\theoremstyle{plain}
\newtheorem{theorem}{Theorem}[section]
\newtheorem{proposition}[theorem]{Proposition}

\theoremstyle{definition}
\theoremstyle{remark}

\title{Channel-Preserving Representation \\ Alignment for EEG-to-Music Reconstruction}

\author{Jiaxin Qing \\
University of California, Berkeley \\
\texttt{jxqing@berkeley.edu} \And
Junwei Lu \\
Harvard University \\
\texttt{junweilu@hsph.harvard.edu} \And
Lexin Li \\
University of California, Berkeley \\
\texttt{lexinli@berkeley.edu}}

\iclrfinalcopy
\begin{document}
 
\maketitle
\begin{abstract}
Reconstructing naturalistic music from EEG requires extracting music-specific information from weak signals distributed across electrodes. Existing spectrogram regression and diffusion approaches provide audio synthesis mechanisms, but learning the EEG-to-music mapping faces an information--estimation tradeoff. Combining electrode measurements into fewer features can ease estimation from limited recordings while obscuring distinctions between music segments. We propose channel-preserving representation alignment, separating electrode information retention from prediction regularization. Per-electrode tokenization preserves spatial identity for learned aggregation, while multi-view self-distillation and channel dropout encourage consistency across electrode subsets during pretraining and music alignment. The aligned representation conditions a pretrained audio generator. Our theory explains these complementary choices. A finite-sample regression analysis identifies when retained task information outweighs estimation cost. For nonlinear alignment, an exact decomposition reveals a prediction-consistency penalty in the masked contrastive objective, with margin-dependent bounds connecting prediction variability to ranking errors across electrode subsets. On NMED-T/H, our method improves within-subject 50-way identification from 0.388 to 0.485 and generated-audio CLAP similarity from 0.635 to 0.683 compared with previous work.
\end{abstract}

\section{Introduction}
    
Brain-computer interfaces (BCIs) decode stimuli from functional magnetic resonance imaging (fMRI) and electroencephalography (EEG). Deep learning and generative models have advanced visual and language reconstruction from brain signals~\citep{eegdecoding1,mindeye2,tang2023semantic,mindvideo,takagi2023high}.

We study reconstructing naturalistic music from EEG recorded during listening. Music combines rhythm, melody, timbre, and emotion over time, while EEG observations are weak, distributed across electrodes, and sensitive to artifacts and channel variability.
    
Existing approaches use spectrogram regression~\citep{eeg2mel} or EEG-conditioned latent diffusion~\citep{postolache2025naturalistic} to generate audio. The learning challenge is to retain music-relevant spatial cues without fitting noise in limited paired recordings. Compression can ease estimation but discard these cues, whereas preservation can increase estimation cost. We therefore preserve electrode identity for learned aggregation while encouraging consistent music predictions across electrode subsets of the same segment, discouraging excessive dependence on a particular electrode configuration.

We address this challenge through channel-preserving representation alignment. Building on per-channel EEG modeling and self-supervision~\citep{labram,eegdino,tfmtokenizer,codebrain}, we combine three components around information retention and prediction consistency. First, \emph{channel-wise tokenization} preserves the spatial origin of temporal patches for learned cross-channel aggregation. Second, \emph{multi-view self-distillation} matches local student views to global teacher views, encouraging shared representations across temporal crops and electrode subsets before paired training. Third, \emph{channel-wise data augmentation} combines synchronized EEG--music crops with channel dropout during alignment, extending consistency learning to music predictions. The aligned representation conditions frozen AudioLDM through CLAP.

Our theory explains the complementary roles of preservation and masking (Section~\ref{sec:identifiability}). In a regression model, preservation lowers prediction risk exactly when the task information lost by compression exceeds the additional estimation cost. For nonlinear predictors, the masked contrastive objective contains an exact prediction-consistency penalty, and mask variability separates into electrode sensitivity and residual interactions. Candidate-margin bounds connect prediction risk and mask variability to ranking stability. These results give a statistical criterion for retaining information and an objective-level explanation of how masking penalizes inconsistent use of that information. Our contributions are as follows.

\begin{itemize}
\setlength{\itemsep}{3pt}
\item \textbf{Channel-preserving EEG--music alignment.} We retain electrode identity before learned aggregation and combine multi-view self-distillation with channel augmentation to learn music-aligned representations across partial observations.

\item \textbf{A theoretical account of preservation and consistency.} A finite-sample regression threshold characterizes when retaining information helps prediction, while an exact nonlinear mask decomposition and margin-dependent bounds connect electrode sensitivity to ranking stability.

\item \textbf{Controlled evidence for the design.} Under matched NMED pretraining and alignment conditions, our method improves within-subject 50-way identification and generated-audio CLAP similarity over CBraMod. Component ablations assess the contributions of tokenization, pretraining, and channel dropout.
\end{itemize}

\section{Related Work}
\label{sec:related}
    
\noindent
\textbf{Brain Decoding.}
Brain decoding has advanced rapidly in vision and language \citep{eegdecoding1, mindeye2, tang2023semantic, mindvideo, takagi2023high}. In vision, \citet{takagi2023high} map fMRI signals into the latent space of Stable Diffusion, while \citet{mindeye2} employ CLIP-style contrastive learning to align fMRI with image representations. For EEG-based image decoding, \citet{eegdecoding1} also adopt CLIP-based alignment. A smaller body of work focuses on pretraining brain encoders prior to alignment, including \citet{mindvis} for fMRI and \citet{bai2024dreamdiffusion} for EEG. Our work follows this encoder-first paradigm, but investigates whether it can support a continuous music target from EEG, whose perceptual content depends on temporal structure as well as high-level semantics.
    
\noindent
\textbf{EEG-to-Music Reconstruction.}
Naturalistic music EEG datasets support the study of neural responses to continuous music~\citep{nmedh,nmedt}. EEG2Mel~\citep{eeg2mel} predicts mel spectrograms for waveform reconstruction. \citet{postolache2025naturalistic} condition a pretrained AudioLDM2 generator on scalp EEG through ControlNet, demonstrating latent diffusion reconstruction on NMED-T. Our work focuses on electrode-preserving representation learning and explicit EEG--music alignment before generation. We evaluate both alignment and generated audio.
    
\noindent
\textbf{EEG Encoders.}
Transformer-based EEG encoders learn representations from large EEG corpora \citep{neurogpt,labram,eegdino,wang2024eegpt,wang2024cbramod,wang2025eegmamba,wang2026deeperbrain}. LaBraM uses per-channel temporal patches and masked neural-code prediction, while EEG DINO combines per-electrode temporal patches with channel-subset multi-view self-supervision. EEGPT and CBraMod learn spatiotemporal representations for downstream decoding. More recently, TFM-Tokenizer learns discrete time-frequency motifs from single-channel EEG \citep{tfmtokenizer}, CodeBrain combines a decoupled tokenizer with a multi-scale architecture \citep{codebrain}, and CSBrain aggregates features across temporal windows and brain regions \citep{csbrain}. We investigate how electrode-level representations support EEG--music alignment and generation through tokenization ablations and matched NMED pretraining comparisons.

\section{Channel-Preserving Representation Alignment}
\label{sec:design}
    
\begin{figure}[t!]
\centering
\includegraphics[width=\linewidth]{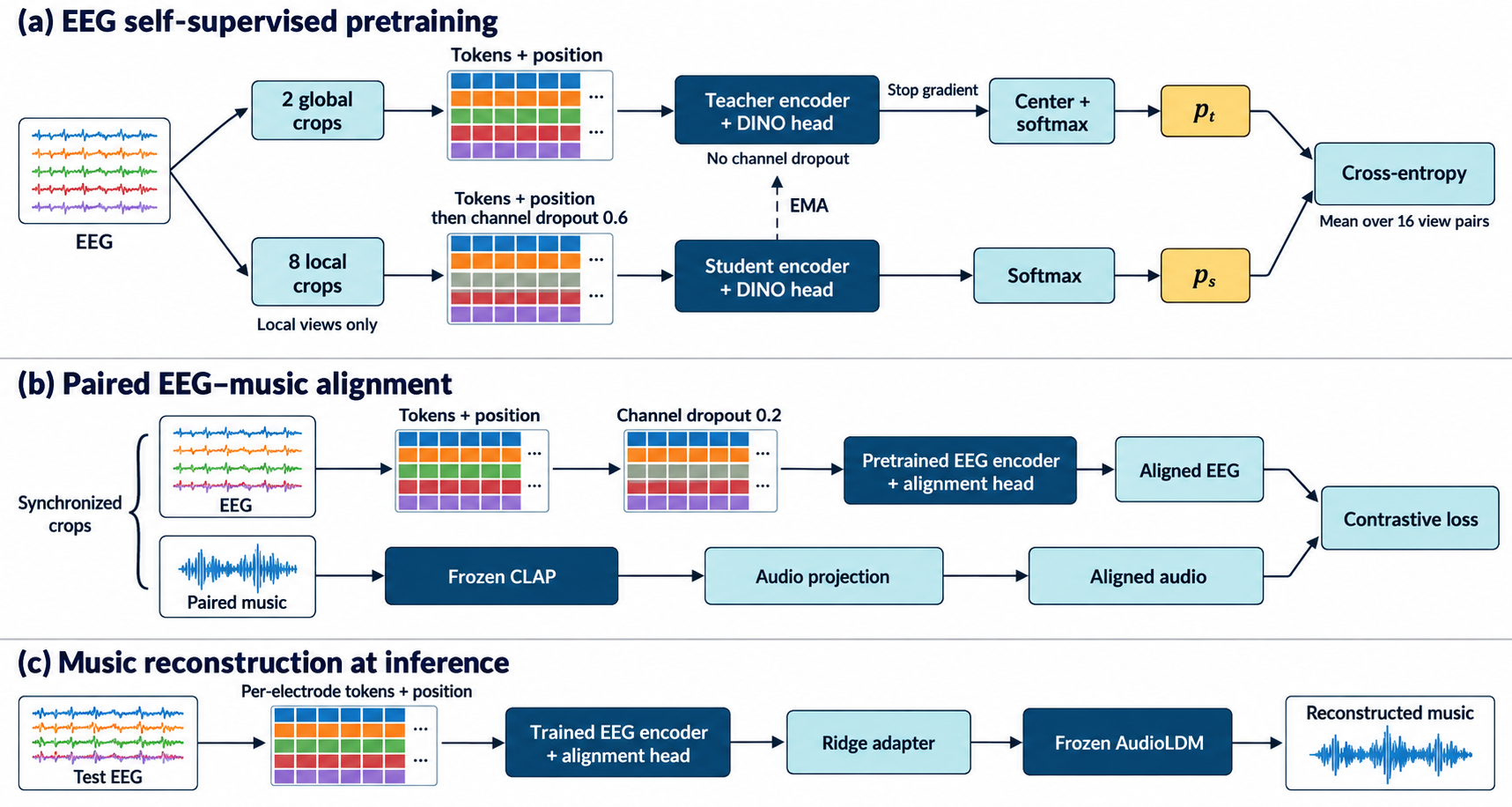}
\caption{Channel-preserving EEG-to-music pipeline. (a) The teacher uses two global views without channel dropout, while the student uses eight local views with channel dropout 0.6 after positional embedding. Cross-entropy averages 16 view pairs, and EMA updates the teacher. (b) Paired alignment uses synchronized crops, frozen CLAP, and EEG channel dropout 0.2. (c) Aligned EEG passes through a ridge adapter to condition frozen AudioLDM at inference. Colors distinguish electrodes; gray rows denote omitted channels.}
\label{fig:eeg_encoder}
\vspace{-12pt}
\end{figure}

Our pipeline combines EEG encoding, alignment to a pretrained CLAP audio encoder~\citep{clap,audioldm}, and reconstruction with pretrained AudioLDM~\citep{audioldm}. The EEG encoder uses channel-wise tokenization and multi-view self-distillation, followed by contrastive alignment with channel augmentation (Figure~\ref{fig:eeg_encoder}).

\subsection{Channel-wise Tokenization, Multi-view Self-distillation, and Data Augmentation}
    
\noindent
\textbf{Channel-wise tokenization.} 
We tokenize each electrode separately and use transformer attention to learn cross-channel interactions. This preserves each measurement's spatial origin and temporal structure before global aggregation.
    
Specifically, each trial is divided into temporal windows of length $t$ and stride $s$. At the 125 Hz sampling rate, each window contains $125t$ time samples. Within a window, each electrode signal is padded as needed and partitioned into 50-sample temporal patches with positional embeddings. The transformer attends over the channel--time tokens and outputs $n+1$ latent tokens, comprising $n$ channel--time patches and one CLS token. The CLS token feeds the self-distillation head during pretraining. During alignment, the temporal convolutional head aggregates the full output token sequence, including CLS, into the shared music representation.
    
\noindent
\textbf{Channel-wise multi-view self-distillation.} 
DINO-style multi-view self-distillation~\citep{dino} pretrains the EEG encoder for consistency across temporal scales and electrode subsets before paired alignment.

Each EEG segment produces two global temporal crops spanning 50--90\% of the segment and eight local crops spanning 10--50\%. The teacher processes global views without channel dropout. The student processes only local views with channel dropout ratio 0.6. Both use the same architecture with independent initialization. Cross-entropy is averaged over all 16 global--local view pairs, and an exponential moving average of the student updates the teacher. Algorithm~\ref{alg:pretraining} gives the procedure, with EMA indexing in Appendix~\ref{app:ssl_details}.

\begin{algorithm}[ht]
\caption{Channel-wise multi-view self-distillation.}
\label{alg:pretraining}
\small
\begin{algorithmic}[1]
            \STATE \textbf{Input:} EEG dataset $\mathcal{D}$, fixed training steps $T$
            \STATE \textbf{Initialize:} independent student $\theta_s$ and teacher $\theta_t$ parameters
            \FOR{step $=1,\ldots,T$}
                \STATE Sample mini-batch $X$ from $\mathcal{D}$
                \STATE Generate two global crops $V_g$ and eight local crops $V_l$ for each $x\in X$
                \STATE Compute $p_s(v)$ with channel dropout 0.6 and stop-gradient $p_t(u)$ without dropout
                \STATE $\mathcal{L} \leftarrow \frac{1}{|V_g||V_l|}\sum_{u\in V_g}\sum_{v\in V_l}\mathrm{CE}\!\left(p_t(u),p_s(v)\right)$
                \STATE $\theta_s \leftarrow \theta_s - \eta \nabla_{\theta_s} \mathcal{L}$, \quad $\theta_t \leftarrow \lambda_j \theta_t + (1 - \lambda_j) \theta_s$ 
            \ENDFOR
\end{algorithmic}
\end{algorithm}
    
\noindent
\textbf{Channel-wise data augmentation.} 
During alignment, EEG and music use synchronized random resized crops with shared start and scale, followed by Gaussian perturbation. EEG channel dropout removes spatial measurements by dropping electrode tokens. After positional embedding, we uniformly retain $\lfloor0.8C\rfloor$ of the $C$ electrodes without replacement and remove all temporal patches of the others (Algorithm~\ref{alg:eeg_augment}).
    
\begin{algorithm}[ht]
\caption{Channel-wise data augmentation.}
\label{alg:eeg_augment}
\small
\begin{algorithmic}[1]
        \REQUIRE data $x \in \mathbb{R}^{C \times L}$, crop scale $p$, noise level $\sigma$
        \ENSURE Augmented channel--time tokens $Z' \in \mathbb{R}^{n' \times d}$
        \STATE $L_c \gets \lfloor pL \rfloor$, \quad $i \sim \mathcal{U}(0,L-L_c)$
        \STATE $\tilde{x} \gets \mathrm{Interpolate}(x[:,i:i+L_c], \mathrm{size}=L)$
        \STATE $Z \gets \mathrm{Tokenize}(\tilde{x}+\epsilon)+\mathrm{Position}$, \quad $\epsilon \sim \mathcal{N}(0,\sigma^2 I)$
        \STATE \textbf{return} $Z' \gets \mathrm{ChannelDrop}(Z,0.2)$, keeping all patches of $\lfloor0.8C\rfloor$ uniformly selected electrodes
\end{algorithmic}
\end{algorithm}

\subsection{Music Encoder, Alignment, and Reconstruction}

\noindent
\textbf{Music encoder.} 
Frozen CLAP~\citep{clap,audioldm} encodes audio segments synchronized with the EEG windows. Its embeddings are projected into the shared alignment space and provide the music representation used to condition AudioLDM.
    
\noindent
\textbf{Alignment.} 
A temporal convolution, mean pooling, and linear layers aggregate the full EEG token sequence, including CLS and channel--time tokens. For batch size $B$ and embedding dimension $D$, the resulting $Z_{eeg}\in\mathbb{R}^{B\times D}$ and $Z_{audio}\in\mathbb{R}^{B\times D}$ are aligned by symmetric cross-entropy on the normalized, temperature-scaled similarity matrix $Z_{eeg}Z_{audio}^\top\in\mathbb{R}^{B\times B}$, with EEG channel dropout across electrode subsets.
    
\noindent
\textbf{Music decoder.} 
A ridge adapter fitted on training pairs with regularization $\lambda=0.1$ maps aligned EEG embeddings into CLAP space to condition frozen AudioLDM~\citep{audioldm} at inference (Figure~\ref{fig:recon}). For encoder comparisons sharing these reconstruction components, differences reflect EEG representation and alignment under a fixed generator.
    
\begin{figure}[t]
\centering
\includegraphics[width=0.65\linewidth]{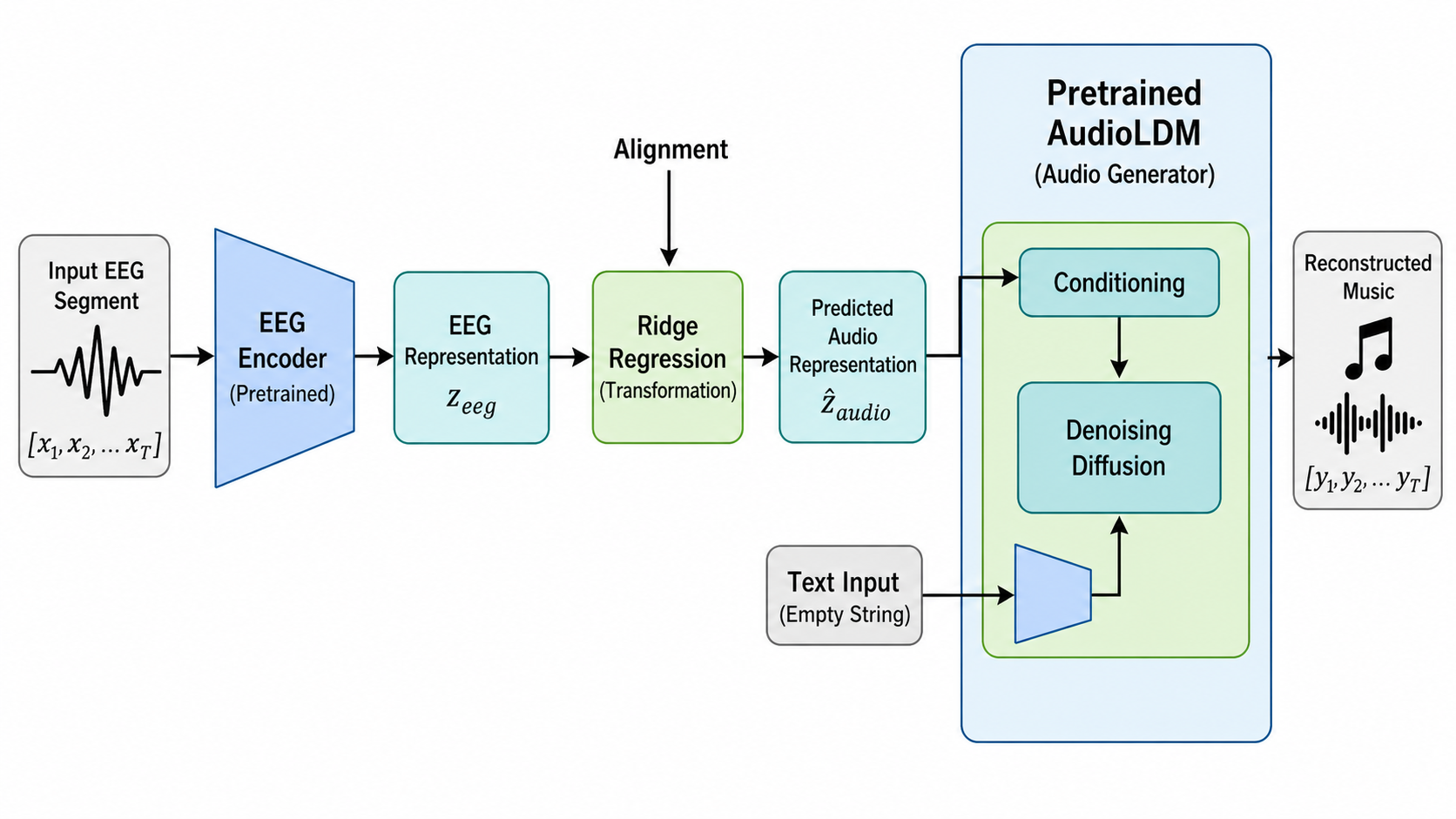}
\caption{The reconstruction pipeline. After alignment, EEG representations are mapped to the CLAP audio embedding space using ridge regression. The transformed embedding is then used to condition AudioLDM for music reconstruction.}
\label{fig:recon}
\end{figure}

\section{Statistical Tradeoffs and Masking Robustness}
\label{sec:identifiability}

Channel preservation and electrode masking address two sources of music-ranking error. Preserving coordinates can retain music-relevant information at an increased estimation cost, while the masked contrastive objective penalizes prediction variability across electrode subsets. We first quantify this statistical tradeoff in a regression model, then decompose the masking penalty for nonlinear encoders, and finally relate both quantities to candidate separation. All proofs are in Appendix~\ref{app:proofs}.

\subsection{Information Retention Versus Estimation Cost}

A regression model isolates when retained task information outweighs the cost of estimating more coordinates. Consider whitened features $H\sim\mathcal N(0,I_p)$ and a continuous music target $V=W_*^\top H+\varepsilon$, with independent, zero-mean noise satisfying $\mathbb E\|\varepsilon\|^2=\sigma^2$. We compare zero-intercept ordinary least squares (OLS) on $H$ and on $AH$, where $A\in\mathbb R^{r\times p}$ has rank $1\leq r<p$ and is fixed independently of the $n$ training pairs. The projection $P_A=A^\top(AA^\top)^{-1}A$ identifies the retained subspace, and $b_A=\|(I-P_A)W_*\|_F^2$ measures the music-relevant energy discarded by compression.

\begin{proposition}
\label{prop:compression_crossover}
For $n>p+1$ independent training pairs and a fresh test pair, the expected squared prediction risks, including test noise, are
\begin{equation}
R_{\mathrm{full}}=\sigma^2\frac{n-1}{n-p-1},\qquad
R_A=(\sigma^2+b_A)\frac{n-1}{n-r-1}.
\label{eq:compression_crossover}
\end{equation}
Consequently,
\begin{equation}
R_{\mathrm{full}}<R_A
\quad\Longleftrightarrow\quad
b_A>\frac{\sigma^2(p-r)}{n-p-1}.
\label{eq:compression_threshold}
\end{equation}
\end{proposition}

Equation~\eqref{eq:compression_threshold} makes the tradeoff explicit. Within this model, compression can reduce prediction risk at smaller sample sizes even when it discards some signal, whereas preservation helps once $b_A$ exceeds the estimation cost. Compression preserves all task-relevant signal when its retained subspace contains the columns of $W_*$, giving $b_A=0$. This random-design analysis~\citep{hsu2012random} isolates the information--estimation tradeoff for OLS with independent training pairs and an unconstrained target. Appendix~\ref{app:compression_proof} explains the relationship between this statistical model and the learned representation.

\subsection{Electrode Masking as Prediction Regularization}

The masked contrastive objective admits an exact decomposition into a mean-logit loss and a prediction-consistency penalty. Fix a paired batch of size $B\geq2$, temporal crops, and model parameters. Let $S_{M,O}$ be the temperature-scaled cosine-logit matrix under electrode selection $M$ and retained-token ordering $O$, with $\bar S=\mathbb E S_{M,O}$. All expectations average the joint draw of $(M,O)$. Write $\ell(s,y)=\log\sum_j e^{s_j}-s_y$ and $\mathcal L(S)=(2B)^{-1}\sum_i[\ell(S_{i,:},i)+\ell(S_{:,i},i)]$. For each row or column $v$, define $J_v=\mathbb E\operatorname{KL}(\operatorname{softmax}(\bar S_v)\Vert\operatorname{softmax}((S_{M,O})_v))$. The log-sum-exp Bregman identity~\citep{banerjee2005bregman} yields
\begin{equation}
\mathbb E\mathcal L(S_{M,O})=\mathcal L(\bar S)+\frac1{2B}\sum_v J_v.
\label{eq:mask_kl}
\end{equation}
The penalty measures disagreement with the mean-logit predictor, connecting electrode masking to adaptive dropout regularization~\citep{wager2013dropout}. This reference generally differs from the full-channel predictor, and optimization changes both terms in Equation~\eqref{eq:mask_kl}.

To identify which electrode effects contribute to this penalty, consider one EEG-to-audio row $s_{M,O}$ and denote its penalty by $J$. Let $M\in\{0,1\}^N$ uniformly retain $k$ of $N$ electrodes and set $\rho=k/N$. Centering the logits as $x_{M,O}=C_Bs_{M,O}$, with $C_B=I-\mathbf1\mathbf1^\top/B$, removes common shifts that leave softmax unchanged. Define $\alpha=\rho(1-\rho)N/(N-1)$ and the first-order electrode coefficients $a_i=\alpha^{-1}\mathbb E[(M_i-\rho)x_{M,O}]$.

\begin{theorem}
\label{thm:mask_projection}
For $1\leq k<N$, the residual $r_{M,O}=x_{M,O}-\mathbb E x_{M,O}-\sum_i(M_i-\rho)a_i$ is orthogonal to every centered mask coordinate, and
\begin{equation}
\mathcal V:=\mathbb E\|x_{M,O}-\mathbb E x_{M,O}\|^2
=\alpha\sum_i\|a_i\|^2+\mathbb E\|r_{M,O}\|^2.
\label{eq:nonlinear_projection}
\end{equation}
The row penalty satisfies $J\leq\mathcal V/4$. If every softmax coordinate along each segment from $\mathbb E s_{M,O}$ to $s_{M,O}$ is at least $\mu>0$, then also $J\geq\mu\mathcal V/2$.
\end{theorem}

The decomposition separates first-order electrode sensitivity from residual mask interactions and ordering variability without requiring an additive encoder. Additivity corresponds to $r_{M,O}=0$, while the residual explicitly retains effects beyond that approximation. The variance controls an upper bound on $J$, and the lower bound makes the relationship two-sided when probabilities remain nonsaturated. At low temperature, the lower bound can be weak. Appendix~\ref{app:mask_projection} illustrates how joint block selection can leave within-block cancellation directions unpenalized in the additive case.

\subsection{From Prediction Error to Music Ranking}

The effect of prediction variability on ranking depends on candidate separation. For a fixed candidate set, suppose $p=\operatorname{softmax}(\mathbb E s_{M,O})$ ranks the paired candidate $y$ uniquely first. The smallest KL cost of reaching an incorrect or tied ranking is
\begin{equation}
c_y(p)=\min_{j\ne y}\left[p_y\log\frac{2p_y}{p_y+p_j}
+p_j\log\frac{2p_j}{p_y+p_j}\right]>0.
\label{eq:ranking_cost}
\end{equation}
A constrained KL minimization, as in classification calibration~\citep{pires2016calibration}, yields the nonasymptotic bound
\begin{equation}
\Pr_{M,O}\!\left[\max_{j\ne y}(s_{M,O})_j\geq(s_{M,O})_y\right]
\leq\min\!\left\{1,\frac{J}{c_y(p)}\right\}
\leq\min\!\left\{1,\frac{\mathcal V}{4c_y(p)}\right\}.
\label{eq:mask_ranking}
\end{equation}
Equation~\eqref{eq:mask_ranking} connects electrode sensitivity and nonlinear residuals to retrieval stability through the separation cost $c_y(p)$. It bounds ranking failures across masked views when the mean-logit predictor is correct. Appendix~\ref{app:ranking} connects the two predictors through an explicit drift bound.

The regression risk connects to ranking through a complementary geometric margin. For any nonzero target $V$, normalized candidates with $u_y=V/\|V\|$, and prediction $\widehat V$, define $\Gamma=\|V\|^2\min_{j\ne y}(1-u_y^\top u_j)$. Ranking by $\widehat V^\top u_j$ and counting ties as errors gives, for every $\gamma_0>0$,
\begin{equation}
\Pr(\mathrm{error})\leq\Pr(\Gamma<\gamma_0)
+\frac{2\mathbb E\|\widehat V-V\|^2}{\gamma_0}.
\label{eq:prediction_ranking}
\end{equation}
The bound separates small-margin cases from errors controlled by prediction risk. For unit targets, $\Gamma$ reduces to the cosine margin. Under the regression model and a common target and candidate distribution, Equation~\eqref{eq:compression_threshold} therefore identifies when preservation tightens the upper bound in Equation~\eqref{eq:prediction_ranking}. Together, the analyses yield two complementary design criteria. Retained task information must justify its estimation cost, and prediction variability across electrode subsets must be assessed relative to candidate separation. These criteria explain how channel preservation and masked alignment address information retention and prediction stability, respectively. The tokenization and dropout ablations in Section~\ref{sec:ablation} evaluate the corresponding design choices through identification accuracy.

\section{Numerical Experiments}

\subsection{Datasets}
\label{sec:data_split}

We use NMED-T \citep{nmedt} and NMED-H \citep{nmedh}, comprising 29.4 hours of EEG from 68 subjects. NMED-T contains 20 subjects and 10 pop songs, and NMED-H contains 48 subjects and four Hindi songs. We use original and reversed NMED-H stimuli. The data provide 125 channels at 150~Hz after the benchmark preprocessing, including 0.3--50~Hz filtering, artifact removal, and per-channel normalization. EEG is downsampled to 125~Hz for our experiments.

\textbf{Evaluation split.} We first divide recordings into non-overlapping one-second EEG--music pairs, then randomly assign 95\% to training and 5\% to testing. This evaluates held-out segments from observed subjects and song collections. Temporally adjacent segments can occur in different partitions. Paired alignment and ridge fitting use training pairs, and both pretraining and alignment use the final training checkpoint. Dropout, crop ranges, and ridge regularization are fixed before evaluation.

\textbf{Pretraining protocols.} The reference experiments use label-free transductive pretraining on all EEG, with eight-second windows and a 6.4-second stride. Appendix~\ref{app:matched} reports a control restricted to training EEG. 

\subsection{Baseline Methods and Evaluation Metrics}

\noindent
\textbf{Linear EEG Reference.}
We fit a linear ridge regression from raw EEG features to CLAP music embeddings and use the same AudioLDM reconstruction pipeline. This reference assesses the contribution of learned EEG representations. 

\noindent
\textbf{Audio Reconstruction Reference.}
We feed the ground-truth music directly into AudioLDM. This reference method measures the degradation introduced by the fixed generator without any EEG prediction, and therefore serves as an audio-input reference for the reconstruction pipeline.

\noindent
\textbf{Music reconstruction baselines.}
EEG2Mel~\citep{eeg2mel} maps EEG to mel spectrograms, followed by Griffin--Lim waveform reconstruction. \citet{postolache2025naturalistic} use EEG-conditioned ControlNet with pretrained AudioLDM2. 

\noindent
\textbf{EEG foundation models.}
Released LaBraM~\citep{labram}, EEGPT~\citep{wang2024eegpt}, and CBraMod~\citep{wang2024cbramod} checkpoints use a common projection head and our ridge/AudioLDM reconstruction pipeline. Our main model uses a temporal head. Table~\ref{tab:controls} controls pretraining conditions for ours and CBraMod and evaluates continued NMED pretraining of CSBrain~\citep{csbrain}. Input adapters and training protocols are detailed in Appendices~\ref{app:comparison_details} and~\ref{app:matched}.

\noindent
\textbf{Evaluation Metrics.}
Identification uses aligned EEG embeddings before generation for our model and EEG encoder baselines, and reconstructed-audio CLAP embeddings for EEG2Mel, Postolache et al., and the audio reference. Cosine similarity retrieves the paired target among 50 candidates or candidates from 14 distinct songs. Both tasks use ten candidate-set samplings. After generation, CLAP similarity compares generated and target audio in the pretrained representation used for conditioning. We also report 10-class genre-label agreement for broad musical style, and SSIM and PSNR for local mel-spectrogram similarity. Appendix~\ref{app:eval_details} details the metrics.

\subsection{Main Results}

Table~\ref{tab:res} reports audio reconstruction metrics for all methods and separates identification into pre-generation alignment (panel A) and audio-output retrieval (panel B). Table~\ref{tab:controls} compares encoders under controlled training conditions. 

\begin{table}[t]
\centering
\caption{Reconstruction and identification on pooled NMED-T/H. Audio-level metrics evaluate generated waveforms in both panels. Identification uses (A) pre-generation EEG embeddings or (B) reconstructed-audio CLAP embeddings ($\dagger$), as defined in Appendix~\ref{app:eval_details}. Our encoder uses transductive SSL. Entries with $\pm$ report mean and SD over ten evaluation samplings.}
\label{tab:res}
\newcommand{\metricstd}[2]{#1{\scriptsize $\pm$ #2}}
\newcommand{\metricstdb}[2]{\textbf{#1}{\scriptsize $\pm$ #2}}
\newcommand{\resultsTableWidth}{1\linewidth}
\setlength{\tabcolsep}{4pt}
\resizebox{\textwidth}{!}{
\begin{tabular}{l c c c c c c}
\toprule
& \multicolumn{4}{c}{\textbf{Audio-level} $\uparrow$} & \multicolumn{2}{c}{\textbf{Identification} $\uparrow$} \\
\cmidrule(lr){2-5}\cmidrule(lr){6-7}
\textbf{Method} & {\color{gray}SSIM} & {\color{gray}PSNR} & CLAP & \shortstack{10-class \\ genre agr.} & \shortstack{50-way \\ id.} & \shortstack{14-way \\ id.} \\
\midrule
\multicolumn{7}{l}{\textbf{A. Alignment retrieval (before generation)}} \\
Linear EEG Reference                         & {\color{gray}0.491} & {\color{gray}14.32} & 0.576 & \metricstd{0.126}{0.067} & \metricstd{0.018}{0.057} & \metricstd{0.067}{0.057} \\
LaBraM~\citep{labram}                        & {\color{gray}0.422} & {\color{gray}14.21} & 0.657 & \metricstd{0.162}{0.056} & \metricstd{0.380}{0.057} & \metricstd{0.681}{0.051} \\
EEGPT~\citep{wang2024eegpt}                  & {\color{gray}0.378} & {\color{gray}14.34} & 0.625 & \metricstd{0.141}{0.065} & \metricstd{0.326}{0.069} & \metricstd{0.643}{0.068} \\
CBraMod~\citep{wang2024cbramod}              & {\color{gray}0.418} & {\color{gray}14.42} & 0.641 & \metricstd{0.169}{0.059}  & \metricstd{0.402}{0.058} & \metricstd{0.690}{0.064} \\
\textbf{Ours}                                & {\color{gray}0.488} & {\color{gray}14.41} & \textbf{0.683} & \metricstdb{0.203}{0.055} & \metricstdb{0.487}{0.067} & \metricstdb{0.692}{0.062} \\
\midrule
\multicolumn{7}{l}{\textbf{B. Audio-output retrieval (after generation)}} \\
Audio Reconstruction Reference$^\dagger$    & {\color{gray}0.416} & {\color{gray}13.26} & 0.752 & \metricstd{0.276}{0.060} & \metricstd{0.598}{0.069} & \metricstd{0.775}{0.061} \\
EEG2Mel$^\dagger$                           & {\color{gray}0.762} & {\color{gray}24.37} & 0.588 & \metricstd{0.132}{0.051} & \metricstd{0.259}{0.050} & \metricstd{0.478}{0.062} \\
\citet{postolache2025naturalistic}$^\dagger$ & {\color{gray}0.469} & {\color{gray}14.52} & 0.672 & \metricstd{0.187}{0.052} & \metricstd{0.392}{0.060} & \metricstd{0.688}{0.061} \\
\bottomrule
\end{tabular}
}
\end{table}

\noindent
\textbf{The main gains are CLAP similarity and EEG--music alignment.}
Our approach reaches CLAP 0.683 and genre agreement 0.203, versus 0.672 and 0.187 for \citet{postolache2025naturalistic}, the next-highest EEG-based method on both metrics.

\noindent
\textbf{Identification in the alignment space.}
In Table~\ref{tab:res}A, our model achieves 0.487 versus CBraMod's 0.402 on 50-way identification before generation, while 14-way scores are close at 0.692 and 0.690. Table~\ref{tab:res}B evaluates retrieval from reconstructed audio. Postolache et al. score 0.392 and 0.688, compared with EEG2Mel's 0.259 and 0.478, respectively. 

\noindent
\textbf{Spectrogram metrics capture a complementary aspect of reconstruction.}
SSIM and PSNR assess local spectral fidelity, complementing CLAP and genre agreement. EEG2Mel leads both metrics among EEG-based methods (Table~\ref{tab:res}).

\begin{table}[t]
\centering
\caption{Encoder comparisons (mean $\pm$ sampling SD over ten evaluation samplings). The one-second comparison matches NMED training data, 40,000 SSL steps, and downstream alignment head. CSBrain uses its released checkpoint followed by NMED training-only SSL.}
\label{tab:controls}
\small
\begin{tabular}{llrrr}
\toprule
Setting & Encoder & 50-way $\uparrow$ & 14-way $\uparrow$ & CLAP $\uparrow$ \\
\midrule
1-s matched NMED SSL & CBraMod & $0.388\pm 0.048$ & $0.68\pm 0.042$ & $0.635\pm 0.050$ \\
 & Ours & $0.485\pm 0.055$ & $0.690\pm 0.052$ & $0.683\pm 0.062$ \\
\midrule
NMED-adapted checkpoint & CSBrain & $0.392\pm 0.046$ & $0.689\pm 0.049$ & $0.640\pm 0.051$ \\
\bottomrule
\end{tabular}
\end{table}

\textbf{Matched pretraining.}
Table~\ref{tab:res} uses released LaBraM, EEGPT, and CBraMod checkpoints with different pretraining data and alignment heads. We therefore train ours and CBraMod from random initialization on identical one-second NMED training EEG for 40,000 steps, matching the alignment head and fine-tuning protocol while retaining each encoder's native SSL objective. Ours reaches 0.485 on 50-way identification and 0.683 on CLAP, versus CBraMod's 0.388 and 0.635. Continued NMED SSL from CSBrain's pretrained checkpoint yields 0.392 and 0.640. The 14-way means span 0.68--0.69 (Table~\ref{tab:controls}).

\subsection{Ablation Study}\label{sec:ablation}

To identify which design choices support EEG--music alignment, Table~\ref{tab:ablation} varies tokenization, pretraining, dropout, and alignment adaptation under the reference protocol. Identification is measured before generation.

\begin{table}[t]
\centering
\caption{Component ablations under the reference protocol. Entries report mean $\pm$ SD over ten identification candidate samplings.}
\label{tab:ablation}
\small
\newcommand{\ablationstd}[2]{#1{\scriptsize $\pm$ #2}}
\newcommand{\ablationstdb}[2]{\textbf{#1}{\scriptsize $\pm$ #2}}
\newcommand{\ablationTableWidth}{0.7\linewidth}
\setlength{\tabcolsep}{5pt}
{
\begin{tabular}{l l c c}
\toprule
\textbf{Component} & \textbf{Variant} & \textbf{50-way} $\uparrow$ & \textbf{14-way} $\uparrow$ \\
\midrule
\textbf{Full model} & \textbf{Ours} & \ablationstdb{0.487}{0.053} & \ablationstdb{0.692}{0.067} \\
\midrule
Tokenization
    & Block tokenization & \ablationstd{0.141}{0.064} & \ablationstd{0.406}{0.062} \\
\midrule
Pretraining
    & No multi-view pretraining               & \ablationstd{0.050}{0.067} & \ablationstd{0.155}{0.066} \\
\midrule
Channel dropout
    & No SSL or alignment dropout & \ablationstd{0.411}{0.067} & \ablationstd{0.598}{0.054} \\
\midrule
Temporal modeling
    & Linear head, encoder fixed       & \ablationstd{0.092}{0.065} & \ablationstd{0.191}{0.061} \\
    & Linear head, encoder fine-tuned  & \ablationstd{0.422}{0.051} & \ablationstd{0.612}{0.063} \\
    & Encoder fixed during alignment & \ablationstd{0.179}{0.062} & \ablationstd{0.425}{0.063} \\
\bottomrule
\end{tabular}
}
\end{table}

\noindent
\textbf{Per-electrode tokens outperform five-channel block tokens.}
Concatenating $g=5$ physically ordered electrode signals and linearly projecting each block reduces identification from 0.487 to 0.141 (50-way) and 0.692 to 0.406 (14-way). Grouping reduces token count while keeping the downstream transformer unchanged. This comparison favors per-electrode tokens over the tested block construction.

\noindent
\textbf{Pretraining improves the channel-oriented representation.}
Training the encoder from random initialization reduces 50-way identification from 0.487 to 0.050 and 14-way identification from 0.692 to 0.155, the largest drop in the ablations.

\noindent
\textbf{Channel dropout supports the reconstruction pipeline.}
Removing channel dropout from both SSL pretraining and paired alignment reduces performance from 0.487 to 0.411 on the 50-way task and from 0.692 to 0.598 on the 14-way task. This joint ablation supports channel dropout across the two training stages.

\noindent
\textbf{Joint adaptation of the encoder and alignment head.}
To assess music-specific adaptation, we vary the head and encoder fine-tuning. A fixed encoder with a linear head scores 0.092 and 0.191, rising to 0.422 and 0.612 when fine-tuned. The full temporal head reaches 0.487 and 0.692, while freezing its encoder yields 0.179 and 0.425. These results support jointly adapting the encoder and temporal aggregation to music alignment.

\textbf{Sensitivity to augmentation settings.}
We vary SSL student dropout and crop ranges separately around a shared default configuration (Table~\ref{tab:sensitivity}), holding alignment dropout at 0.2. Lowering student dropout from 0.6 to 0.1--0.3 gives 50-way means of 0.477--0.485, compared with 0.487 at the default. Varying crops yields means of 0.477--0.487. The tested settings give similar means; Appendix~\ref{app:sensitivity} details the protocol.

\begin{table}[t]
\centering
\caption{Sensitivity to SSL augmentation (mean $\pm$ SD over ten identification candidate samplings). Alignment dropout is 0.2. Dropout rows use default crops; crop rows use SSL student dropout 0.6. The two default rows report the same evaluation results. $G$/$L$ denote global/local crop fractions.}
\label{tab:sensitivity}
\small
\begin{tabular}{llcc}
\toprule
Parameter & Setting & 50-way $\uparrow$ & 14-way $\uparrow$ \\
\midrule
SSL student dropout & 0.1 & $0.477\pm0.020$ & $0.689\pm0.023$ \\
 & 0.2 & $0.480\pm0.045$ & $0.688\pm0.049$ \\
 & 0.3 & $0.485\pm0.062$ & $0.691\pm0.051$ \\
 & 0.6 (default) & $0.487\pm0.053$ & $0.692\pm0.067$ \\
\midrule
View crops & Short $G[0.4,0.8],L[0.1,0.4]$ & $0.477\pm0.049$ & $0.681\pm0.060$ \\
 & Default $G[0.5,0.9],L[0.1,0.5]$ & $0.487\pm0.053$ & $0.692\pm0.067$ \\
 & Long $G[0.6,1.0],L[0.2,0.6]$ & $0.485\pm0.049$ & $0.691\pm0.064$ \\
\bottomrule
\end{tabular}
\end{table}

\subsection{Sensor-level Attention Analysis}
\label{sec:interpretability}

Channel-wise tokenization exposes attention at electrode resolution. We average final-layer CLS-to-channel attention over temporal patches, heads, and windows, then standardize subject-level maps across electrodes before averaging across subjects. Figure~\ref{fig:interpretability} shows the resulting song-level scalp maps.

\begin{figure}[t]
\centering
\includegraphics[width=0.9\linewidth]{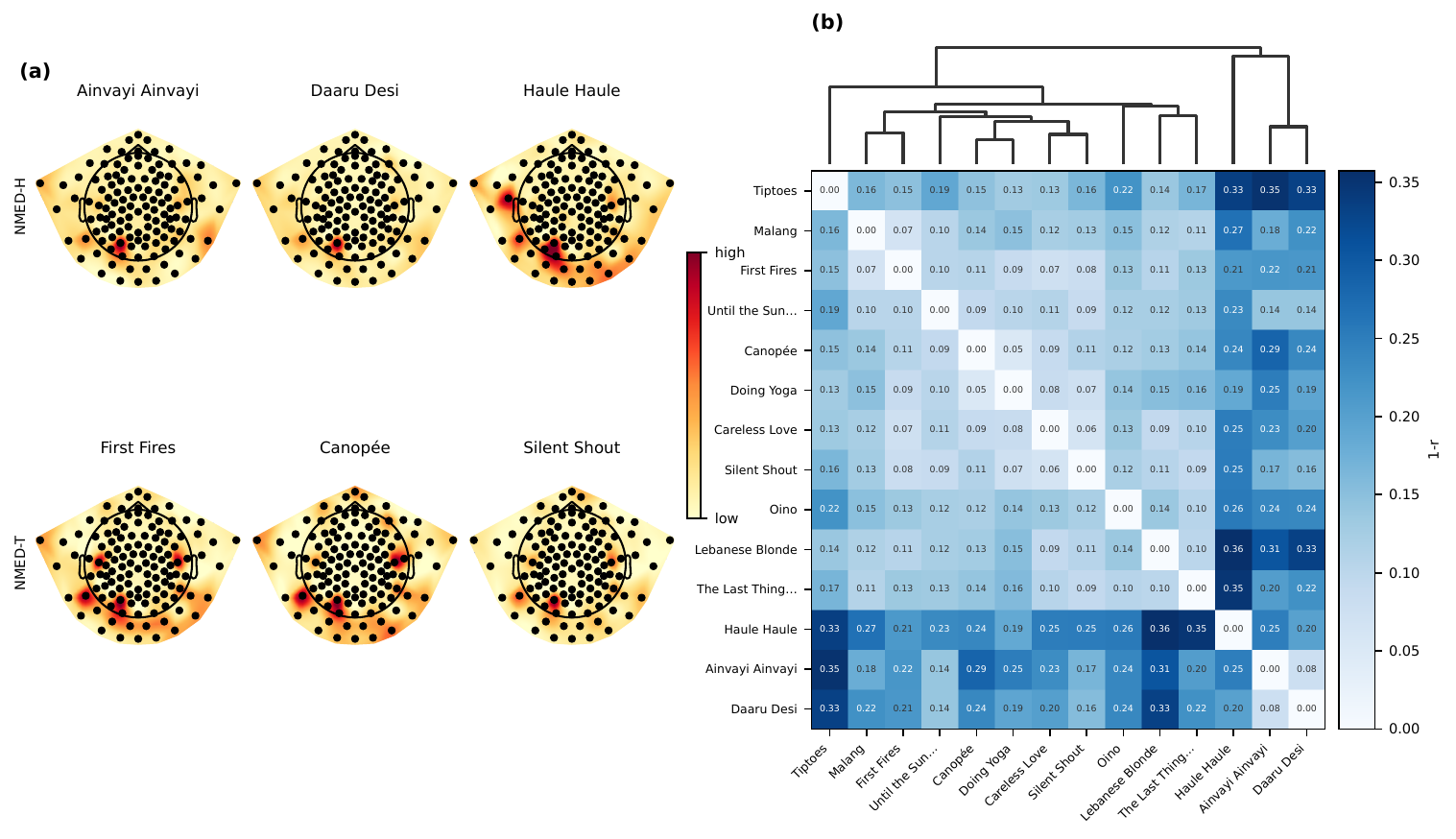}
\caption{Channel attention patterns across songs. (a) Per-song maps for six representative songs from NMED-H (top) and NMED-T (bottom), plotted on the 10-10 scalp layout. (b) Pairwise dissimilarity between the 14 per-song maps as 1 - Pearson correlation.}
\label{fig:interpretability}
\end{figure}

\noindent
\textbf{Song-dependent posterior temporal attention.}
Figure~\ref{fig:interpretability}(a) shows sensor-level attention peaks at lateral posterior temporal and temporo-occipital electrodes, sometimes left-lateralized. Related studies report temporal involvement in naturalistic music~\citep{alluri2012large}, limbic and paralimbic responses to musical affect~\citep{koelsch2014brain}, and EEG tracking of temporal structure and beat~\citep{nmedh}.

\noindent
\textbf{Dataset-associated variation in attention.}
NMED-H songs cluster separately from NMED-T in Figure~\ref{fig:interpretability}(b), with many between-dataset distances above 0.20 and within-NMED-T distances of 0.05--0.15. The datasets contain Hindi and Western pop, respectively. Acoustic content, participant composition, and recording characteristics may contribute to the separation.

\section{Discussion}

Under matched NMED pretraining data, update counts, and alignment heads, channel-preserving alignment improves within-subject 50-way identification and generated-audio CLAP similarity over CBraMod.

\paragraph{Limitations and Broader Impacts.}
Our evaluation uses randomly held-out segments from observed subjects and songs. Subject-held-out, song-held-out, and temporally blocked protocols would test transfer. Larger paired datasets could support speech and environmental sounds, while stronger EEG-to-audio adapters and music encoders could improve reconstruction. 

\clearpage
\subsection*{Ethics Statement}
This study uses the public NMED datasets and concerns music presented during listening. EEG recordings contain sensitive participant information. Applications of neural decoding should use informed consent, appropriate data governance, and explicit safeguards against nonconsensual inference. Generated audio is a model output conditioned on EEG and should be interpreted in the context of the listening task.

\subsection*{AI Use Statement}
Generative AI tools assisted with manuscript revision, figure generation, integration and interpretation of results, literature lookup, formulation and proof writing for the statistical, fixed-count masking, and ranking analyses, clarification of mathematical exposition, and LaTeX formatting. 

\bibliographystyle{iclr2027_conference}

\bibliography{paper}

\newpage
\appendix
\section{Proofs and Implications of the Statistical and Masking Analysis}
\label{app:proofs}

The analysis connects candidate-ranking error to statistical estimation from compressed representations and to prediction variability under electrode selection. The regression route uses independent samples from a Gaussian feature model. The masking route analyzes nonlinear logits at fixed model parameters, with randomness supplied by electrode selection and ordering.

\subsection{Compression, Lost Task Energy, and Finite-Sample Risk}
\label{app:compression_proof}

\paragraph{Population information loss.}
For any square-integrable target $V$ and deterministic compression $T=g(H)$, conditional expectation gives
\[
\mathbb E\|V-f(T)\|^2
=\mathbb E\|V-\mathbb E[V\mid H]\|^2
+\mathbb E\|\mathbb E[V\mid H]-\mathbb E[V\mid T]\|^2
+\mathbb E\|\mathbb E[V\mid T]-f(T)\|^2.
\]
Indeed, the first residual has zero conditional mean given $H$ and the second has zero conditional mean given $T$. All cross terms vanish. The middle term measures information lost about the conditional mean and vanishes for an injective grouping. The finite-sample analysis below adds the estimation cost of learning from each representation.

\paragraph{Regression model and estimator.}
For Proposition~\ref{prop:compression_crossover}, $(H_i,\varepsilon_i)$ are independent copies of $(H,\varepsilon)$, $H\sim\mathcal N(0,I_p)$, $V=W_*^\top H+\varepsilon\in\mathbb R^d$, $\mathbb E\varepsilon=0$, $\mathbb E\|\varepsilon\|^2=\sigma^2<\infty$, and $\varepsilon$ is independent of $H$. Both estimators are zero-intercept OLS, and risk averages over training data and a fresh independent test pair. The matrix $A$ is fixed independently of the training pairs, as stipulated in Proposition~\ref{prop:compression_crossover}. The condition $n>p+1$ ensures existence of the inverse-Gram expectations. Dimensions $p,r$ count feature coordinates in this regression model.

\begin{proof}[Proof of Proposition~\ref{prop:compression_crossover}]
Choose an orthonormal basis $Q\in\mathbb R^{p\times r}$ of the row space of $A$. Since $AH$ and $Q^\top H$ are related by an invertible linear map, zero-intercept OLS fitted values and test predictions agree under the two feature parameterizations. Decompose
\[
V=W_*^\top QQ^\top H+\eta,
\qquad \eta=W_*^\top(I-QQ^\top)H+\varepsilon.
\]
Orthogonal projections of an isotropic Gaussian are independent. Thus $\eta$ is independent of $Q^\top H$, has mean zero, and satisfies
\[
\mathbb E\|\eta\|^2=\sigma^2+\|(I-QQ^\top)W_*\|_F^2=\sigma^2+b_A.
\]
Let $X\in\mathbb R^{n\times r}$ contain the observed features $Q^\top H_i$, and let $\Theta=Q^\top W_*$. Then $\widehat\Theta-\Theta=(X^\top X)^{-1}X^\top E$, where the rows of $E$ are independent copies of $\eta$. Their total conditional covariance satisfies $\mathbb E[EE^\top\mid X]=(\sigma^2+b_A)I_n$. A fresh feature has second moment $I_r$, so its conditional estimation error is
\[
\mathbb E[\|X_{\rm new}^\top(\widehat\Theta-\Theta)\|^2\mid X]
=(\sigma^2+b_A)\operatorname{tr}((X^\top X)^{-1}).
\]
The new residual adds $\sigma^2+b_A$ to prediction risk.

The inverse-Gram expectation can be obtained without an asymptotic approximation. Write $x_j$ for column $j$ of $X$ and $P_{-j}$ for the orthogonal projector onto its other $r-1$ columns. The Schur complement gives
\[
[(X^\top X)^{-1}]_{jj}=\bigl[x_j^\top(I-P_{-j})x_j\bigr]^{-1}.
\]
Conditional on the other columns, the denominator is $\chi^2_{n-r+1}$, because $x_j$ is independent standard Gaussian and $I-P_{-j}$ has rank $n-r+1$ almost surely. Integrating the chi-square density gives $\mathbb E[(\chi^2_m)^{-1}]=1/(m-2)$ for $m>2$. Therefore the expected trace is $r/(n-r-1)$ and
\[
R_A=(\sigma^2+b_A)\left(1+\frac{r}{n-r-1}\right).
\]
Using all $p$ coordinates gives the same expression with $r=p$ and $b_A=0$. This proves Equation~\eqref{eq:compression_crossover}. Since both denominators are positive, comparing the two risks and canceling $n-1$ gives Equation~\eqref{eq:compression_threshold}.
\end{proof}

\paragraph{Correlated electrode coordinates.}
Whitening extends the calculation to correlated Gaussian electrode features through an analytical coordinate change. If raw model features $H_0\sim\mathcal N(0,\Sigma_H)$ with $\Sigma_H\succ0$, and $V=B_*^\top H_0+\varepsilon$, then set $H=\Sigma_H^{-1/2}H_0$, $W_*=\Sigma_H^{1/2}B_*$ and $A=A_0\Sigma_H^{1/2}$ for a raw compression $A_0H_0$. Equivalently,
\[
b_A=\operatorname{tr}\!\left[B_*^\top\left(\Sigma_H-
\Sigma_HA_0^\top(A_0\Sigma_HA_0^\top)^{-1}A_0\Sigma_H\right)B_*\right].
\]
This expression makes the discarded task variance depend on the projection, feature covariance, and target coefficients. Correlation between electrodes is handled by whitening, while dependence between training examples calls for a separate risk calculation.

\paragraph{Sharp example and interpretation.}
For $p=2$, $r=1$, $W_*=w(1,-1)^\top/\sqrt2$, and $A=(1,1)/\sqrt2$, averaging discards the entire signal: $b_A=w^2$. Full coordinates have lower OLS prediction risk exactly when $w^2>\sigma^2/(n-3)$, for $n>3$. In contrast, $A=(1,-1)/\sqrt2$ retains all target-relevant signal, so $b_A=0$ and compression reduces the estimation cost when $\sigma^2>0$. These constructed examples isolate the effect of signal orientation relative to the retained subspace. Applying this risk calculation to the NMED pipeline requires accounting for its learned features, unit CLAP targets, ridge adapter, and dependent temporal windows, which can alter the crossover threshold.

\subsection{Exact Symmetric Contrastive Decomposition}
\label{app:mask_kl_proof}

Let $a(s)=\log\sum_j e^{s_j}$. For any fixed paired batch and parameters, set $\bar s=\mathbb E s_{M,O}$ and $p=\operatorname{softmax}(\bar s)$. The logits have finite moments since the logit scale is fixed and cosine similarities are bounded. Direct substitution gives
\[
\operatorname{KL}(p\Vert\operatorname{softmax}(s_{M,O}))
=a(s_{M,O})-a(\bar s)-p^\top(s_{M,O}-\bar s).
\]
Taking expectation cancels the last term. The linear target term $-(s_{M,O})_y$ in cross-entropy also cancels, proving
\[
\mathbb E\ell(s_{M,O},y)=\ell(\bar s,y)+J.
\]
Applying this identity to all $B$ rows and all $B$ columns, with coefficient $1/(2B)$, proves Equation~\eqref{eq:mask_kl}. The decomposition specializes the log-sum-exp Bregman identity to the symmetric alignment objective. The same argument applies to a fixed soft teacher target since its contribution is linear in the logits. For a DINO view pair, conditioning on the teacher output and temporal crops gives the same decomposition of the student loss at a fixed teacher state.

\paragraph{The reference predictor.}
The reference $\operatorname{softmax}(\mathbb E s_{M,O})$ is the normalized geometric average of the mask-conditioned probabilities. Its relation to the full-channel prediction is quantified by the drift analysis in Appendix~\ref{app:ranking}. The decomposition is exact at each parameter setting, and training jointly changes the mean-logit loss and consistency term.

\subsection{Nonlinear Electrode Effects on the Fixed-Count Slice}
\label{app:mask_projection}

Let $M\in\{0,1\}^N$ be uniform subject to $\sum_iM_i=k$, where $1\leq k<N$ and $N>1$. The full augmentation draw can additionally contain a random ordering $O$ of retained electrodes. Expectations in what follows average over the joint law of $(M,O)$; no assumption on the encoder's invariance to that ordering is made. Center each candidate-logit vector using $C_B=I-\mathbf1\mathbf1^\top/B$, so $x_{M,O}=C_Bs_{M,O}$. This removes the direction in which changing every logit by the same amount leaves the predictive distribution unchanged.

\begin{proof}[Proof of Theorem~\ref{thm:mask_projection}]
Let $z=M-\rho\mathbf1$ with $\rho=k/N$. Uniform subset sampling gives
\[
\mathbb E z=0,\qquad
\mathbb E[zz^\top]=\alpha\left(I-\frac{\mathbf1\mathbf1^\top}{N}\right),
\qquad \alpha=\frac{\rho(1-\rho)N}{N-1}.
\]
This follows from $\mathbb E M_i=\rho$ and $\mathbb E M_iM_j=k(k-1)/(N(N-1))$ for $i\ne j$. Define $a_i=\alpha^{-1}\mathbb E[z_i x_{M,O}]$. Since $\sum_i z_i=0$, we have $\sum_i a_i=0$. Let $r_{M,O}=x_{M,O}-\mathbb E x_{M,O}-\sum_i z_i a_i$. It has mean zero and, for each $j$,
\[
\mathbb E[z_jr_{M,O}]=\alpha a_j-\sum_i\mathbb E[z_jz_i]a_i
=\alpha a_j-\alpha\left(a_j-N^{-1}\sum_i a_i\right)=0.
\]
Consequently the cross-covariance of $r_{M,O}$ with $\sum_i z_i a_i$ vanishes. Expanding the covariance proves the stronger matrix identity
\[
\operatorname{Cov}(x_{M,O})=\alpha\sum_i a_i a_i^\top+\mathbb E[r_{M,O}r_{M,O}^\top].
\]
Taking trace yields Equation~\eqref{eq:nonlinear_projection}. This is a degree-one orthogonal projection on the fixed-size Boolean slice~\citep{filmus2016slice}, with any additional ordering randomness included in the probability space. The residual $r_{M,O}$ retains nonlinear interactions and ordering effects beyond the first-order mask coordinates.

For the KL bounds, set $\delta=s_{M,O}-\bar s$ and $\delta_\perp=C_B\delta$. The Hessian of $a$ is $F_q=\operatorname{diag}(q)-qq^\top$ with $q=\operatorname{softmax}(s)$. It annihilates $\mathbf1$ and has operator norm at most $1/2$, since its absolute row sum is $2q_j(1-q_j)\leq1/2$. Taylor's integral remainder gives
\[
\operatorname{KL}(p\Vert\operatorname{softmax}(s_{M,O}))
=\int_0^1(1-t)\delta_\perp^\top F_{\operatorname{softmax}(\bar s+t\delta)}\delta_\perp\,\mathrm dt
\leq\tfrac14\|\delta_\perp\|^2.
\]
If each probability along the segment is at least $\mu$, then for $v\perp\mathbf1$,
\[
v^\top F_qv=\min_c\sum_jq_j(v_j-c)^2
\geq\mu\min_c\sum_j(v_j-c)^2=\mu\|v\|^2.
\]
The same integral formula therefore gives a lower bound $\mu\|\delta_\perp\|^2/2$. Taking expectation establishes both claimed bounds.
\end{proof}

\paragraph{Non-saturation and the local Fisher limit.}
If all logits are in $[-\beta,\beta]$, one may take $\mu=e^{-2\beta}/B$. This gives a positive lower coefficient at finite logit scale. Its exponential dependence on $\beta$ makes probability non-saturation central to a quantitatively useful bound. A uniform positive lower bound on probabilities is necessary for a positive lower constant independent of the logits. Locally, with $p=\operatorname{softmax}(\bar s)$ and $F_p=\operatorname{diag}(p)-pp^\top$,
\[
J=\frac\alpha2\sum_i\|a_i\|_{F_p}^2
+\frac12\mathbb E\|r_{M,O}\|_{F_p}^2+\mathcal R_3,
\qquad |\mathcal R_3|\leq\tfrac43\mathbb E\|C_B(s_{M,O}-\bar s)\|^3.
\]
To obtain the remainder, apply Taylor expansion to centered logits. The third directional derivative of log-sum-exp is a categorical third central moment, bounded by $8\|C_B\delta\|_\infty^3$. The Taylor coefficient $1/6$ gives the displayed bound. The expansion thus separates first-order electrode sensitivity, nonlinear residual variance, and a controlled third-order remainder.

\paragraph{Cancellation under joint block selection.}
In the explicitly additive special case, write $x_{M,O}=x_0+\sum_iM_i b_i$, with $b_i\perp\mathbf1$. Under channel selection the variance is
\[
\alpha\sum_i\|b_i-\bar b\|^2,\qquad \bar b=N^{-1}\sum_i b_i.
\]
For $K$ equal-sized blocks selected jointly at a fixed count, the corresponding formula instead involves the block sums $B_g=\sum_{i\in g}b_i$, centered across blocks. In particular, $b_i=v$ and $b_j=-v$ for two channels in one block give zero block sum while channel selection has variance $2\alpha\|v\|^2>0$ for $v\ne0$. Thus electrode-subset selection exposes a cancellation direction that joint block selection leaves unchanged. The construction isolates augmentation granularity within a fixed additive representation. Here $b_A$ measures target information discarded by compression, while masking variance measures prediction changes under partial observation.

\subsection{Decision Margins and Retrieval Error}
\label{app:ranking}

\paragraph{A sharp KL cost for changing the top-ranked candidate.}
Fix a candidate set, a correct index $y$, and a strictly positive distribution $p$ with $p_y>\max_{j\ne y}p_j$. For any particular competitor $j$, minimize $\operatorname{KL}(p\Vert q)$ over the simplex subject to $q_j\geq q_y$. Minimizing KL is equivalent to minimizing the strictly convex function $-\sum_l p_l\log q_l$. Its unconstrained minimum $q=p$ violates the constraint, so the constraint is active. With $q_y=q_j=t$, the Lagrange equations for normalization give
\[
q_l=p_l\quad(l\notin\{y,j\}),\qquad t=(p_y+p_j)/2.
\]
For example, the normalization multiplier is one: the equations are $q_l=p_l/\lambda$ for the other coordinates and $2t=(p_y+p_j)/\lambda$, whose sum forces $\lambda=1$. Strict convexity makes this feasible stationary point the unique minimum. Substitution gives the competitor cost in Equation~\eqref{eq:ranking_cost}. Taking the minimum over $j\ne y$ handles the union of the decision-boundary halfspaces without a union-bound factor of $B-1$.

Apply this to $q=\operatorname{softmax}(s_{M,O})$ and $p=\operatorname{softmax}(\mathbb E s_{M,O})$. Every incorrect or tied masked ranking incurs KL at least $c_y(p)$. Hence
\[
J=\mathbb E\operatorname{KL}(p\Vert q_{M,O})
\geq c_y(p)\Pr_{M,O}[\max_{j\ne y}(s_{M,O})_j\geq(s_{M,O})_y],
\]
which proves Equation~\eqref{eq:mask_ranking} using Theorem~\ref{thm:mask_projection}. The bound quantifies ranking stability for a fixed predictor whose mean logits rank the paired candidate uniquely first. Its numerator is the KL penalty for that EEG-to-audio row. An incorrect or tied mean prediction gives zero cost of entering the error region, and the bound becomes informative when the row penalty is smaller than $c_y(p)$.

\paragraph{Full-channel margins and predictor drift.}
Let $p^0=\operatorname{softmax}(s_{\rm full})$, let $\gamma=p^0_y-\max_{j\ne y}p^0_j>0$, and suppose $\|p-p^0\|_\infty\leq\delta<\gamma/2$. Then
\[
p_y-p_j\geq\gamma-2\delta>0\qquad(j\ne y).
\]
Thus $p$ has the same correct top candidate and the preceding theorem applies. A simpler bound follows without evaluating $c_y$: any wrong-ranking $q$ satisfies $\|p-q\|_1\geq p_y-p_j\geq\gamma-2\delta$, and Pinsker's inequality gives
\[
\Pr_{M,O}(\mathrm{error})\leq\min\left\{1,\frac{2J}{(\gamma-2\delta)^2}\right\}.
\]
The bound identifies three quantities for assessing masked retrieval relative to a full-channel reference. The margin $\gamma$ measures candidate separation, the drift $\delta$ measures the change to the mean-mask predictor, and $J$ measures variability around that predictor. Applying the bound requires estimating drift alongside the masking penalty.

\paragraph{Squared prediction error and continuous music targets.}
Let $V\ne0$, let all candidates $u_j$ be unit vectors with $u_y=V/\|V\|$, and rank them by $\widehat V^\top u_j$. For $\widehat V\ne0$, this is the same ranking as cosine similarity; a zero prediction is conservatively counted as a tie/error. Write $e=\widehat V-V$. If candidate $j$ ties or beats $y$, then
\[
\|V\|(1-u_y^\top u_j)
=V^\top(u_y-u_j)
\leq -e^\top(u_y-u_j)
\leq\|e\|\sqrt{2(1-u_y^\top u_j)}.
\]
When $1-u_y^\top u_j>0$, squaring gives $\|e\|^2\geq\|V\|^2(1-u_y^\top u_j)/2\geq\Gamma/2$. Duplicate candidates have $\Gamma=0$ and are handled by the low-margin event. On $\Gamma\geq\gamma_0$, an error therefore implies $\|e\|^2\geq\gamma_0/2$. Markov's inequality yields Equation~\eqref{eq:prediction_ranking}. If $V=0$ is possible, define $\Gamma=0$ on that event; it is included in the low-margin term.

This lemma holds for any joint distribution of targets, candidate sets, and predictors. For normalized surrogate targets, it combines with the OLS risks in Proposition~\ref{prop:compression_crossover}. For actual unit alignment targets, setting $\|V\|=1$ gives the bound in terms of the predictor's mean squared error. Holding the candidate-margin distribution fixed makes a lower prediction risk tighten the upper bound. Ordering these upper bounds does not establish the ordering of actual retrieval errors. The 50-way and 14-way evaluations have distinct candidate distributions, each with its own margin distribution.

\clearpage
\section{Experimental Setup and Implementation Details}
\label{app:pretraining_details}

\subsection{Benchmark Data, Code and Licenses}
The code and data used for the reconstruction component are publicly available at \url{https://github.com/haoheliu/AudioLDM}. AudioLDM code and pretrained weights are governed by the licenses specified in the corresponding release. The NMED datasets and pretrained models used in this work are cited in the main text, and our use follows their public research terms.

\subsection{Data Split and Pretraining Protocols}
\label{app:data_protocols}

\noindent
\textbf{Downstream split and model selection.}
EEG is downsampled to 125 Hz. For paired alignment, we segment recordings into non-overlapping one-second EEG--music pairs with a one-second stride before the random 95\%/5\% training/test assignment. There is no temporal purge gap between adjacent alignment segments. The test set is excluded from paired alignment, ridge fitting, hyperparameter selection, and checkpoint selection. All models use their final training checkpoints.

\noindent
\textbf{Reference and training-only pretraining.}
Both protocols use eight-second EEG windows with a 6.4-second stride, giving 1000 samples per window and 800 samples per stride. Reference SSL uses all EEG without music labels, including EEG reserved for downstream testing. The training-only control excludes any eight-second window sharing even one raw sample with a test segment, while retaining the other reference settings. Both use two global views and eight local views. Table~\ref{tab:pretraining_protocols} summarizes the distinction from the one-second matched comparison. Results and comparison-specific implementation details are given in Appendix~\ref{app:matched}.

\begin{table}[htbp]
\centering
\caption{Pretraining protocols used in the reference experiments and controlled comparisons. The matched protocol uses each encoder's native SSL objective. Matching optimizer updates does not match computational cost.}
\label{tab:pretraining_protocols}
\small
\begin{tabular}{lccc}
\toprule
 & Reference & Training-only control & Matched comparison \\
\midrule
SSL EEG & All EEG & Training EEG only & Training EEG only \\
Window length & 8 s & 8 s & 1 s \\
SSL updates & 30,000 & 30,000 & 40,000 \\
Initialization & Random & Random & Random \\
Encoder(s) & Ours & Ours & Ours / CBraMod \\
Results & Table~\ref{tab:res} & Table~\ref{tab:train_only} & Table~\ref{tab:controls} \\
\bottomrule
\end{tabular}
\end{table}

\subsection{Evaluation Metrics and Statistical Reporting}
\label{app:eval_details}

\noindent
\textbf{Audio-level metrics.}
SSIM and PSNR measure mel spectrogram distance between reconstructed and ground-truth music. They quantify local spectral similarity, which complements the embedding-based CLAP score. The CLAP score is cosine similarity between reconstructed and ground-truth audio embeddings from the pretrained CLAP model~\citep{clap}. CLAP similarity and alignment retrieval measure correspondence in the training target representation space. Genre-label agreement uses the frozen Wav2Vec2 classifier \texttt{dima806/music\_genres\_classification} and its released feature extractor.\footnote{\url{https://huggingface.co/dima806/music_genres_classification}} The checkpoint is trained on GTZAN~\citep{music_class1}, whose ten classes are blues, classical, country, disco, hip hop, jazz, metal, pop, reggae, and rock. We pass the first waveform channel to the feature extractor with a 16-kHz sampling-rate argument and zero-pad reconstructed waveforms shorter than 16,000 samples. For each reference waveform, including reversed stimuli, its classifier argmax supplies the comparison label. The score is the fraction of reconstructed/reference waveform pairs with identical predicted labels, averaged equally over pairs in each evaluation sample. Thus it measures classifier-label consistency between the two audio signals. The reported variability is across ten evaluation resamplings.

\noindent
\textbf{Identification metrics.}
For 50-way identification, we sample 50 paired test examples and form a query--candidate cosine-similarity matrix. Our model and the EEG foundation-model baselines use aligned EEG queries and projected CLAP target-audio candidates before generation. The linear EEG reference uses ridge-predicted CLAP queries and frozen CLAP target-audio candidates. A query is correct when its paired candidate has the largest similarity. The 14-way identification task uses the same rule, with batches formed from 14 different songs. Candidate sets exclude duplicate audio excerpts arising across subjects or listening repeats, so each query has a unique paired target. Both tests are repeated 10 times. For EEG2Mel, Postolache et al., and the audio reconstruction reference, the query is the CLAP embedding of reconstructed audio and candidates are CLAP embeddings of target audio. These scores evaluate audio-output retrieval, whereas the EEG encoder rows evaluate alignment before generation. 

\noindent
\textbf{Sampling variability and training seeds.}
The $\pm$ entries in the reference, component-ablation, and augmentation sensitivity tables are standard deviations over ten identification candidate samplings. For the encoder comparison in Table~\ref{tab:controls}, $\pm$ denotes sampling SD over ten evaluation samplings with trained models held fixed. Identification is evaluated on sampled candidate sets, and CLAP similarity is averaged over sampled generated--target audio pairs. The training-only control reports mean and standard deviation across three training seeds separately.

\subsection{Model Implementation}
\label{app:model_implementation}

\subsubsection{Self-distillation and Teacher Updates}
\label{app:ssl_details}
The projection head uses a three-layer MLP with dimensions $512\!\to\!2048\!\to\!2048\!\to\!256$, GELU activations after the first two layers, and no batch normalization. Its output is $\ell_2$ normalized and mapped to 512 logits by a bias-free, weight-normalized linear layer whose scale is fixed at one. Student and teacher temperatures are 0.1 and 0.04, respectively, with no teacher-temperature warmup. The teacher logits are centered using a running batch mean with momentum 0.9 before temperature scaling and softmax.

Global crops are resized to the pretraining window length and local crops to half that length (1000 and 500 samples in the reference protocol). Teacher and student are initialized independently. The teacher remains in evaluation mode and is evaluated without gradients, so its channel dropout is inactive. The student is trained with channel dropout 0.6. The center starts at zero. For teacher EMA, a cosine array of length $EU$ is constructed as
\[
\lambda_j=1-\frac{1-0.996}{2}\left(1+\cos\frac{\pi j}{EU}\right),
\]
where $U$ is the number of optimizer updates per epoch and $E=\lceil T/U\rceil$ for $T$ training steps. The array is indexed by the zero-based minibatch index $j$ within each epoch, which resets at each epoch boundary. Teacher parameters are updated after the student optimization step.

\subsubsection{Temporal Alignment Head}
The convolutional head applies one 1-D convolution along the full output sequence of CLS and channel--time tokens, with 512 input and output channels, kernel size 3, stride 1, and no padding. Mean pooling over output tokens is followed by a $512\!\to\!512$ linear layer, GELU, a second $512\!\to\!512$ linear layer, and layer normalization.

\subsection{Baseline Method Implementations}
\label{app:comparison_details}

\noindent
\textbf{Shared alignment protocol for foundation-model baselines.}
For the released-checkpoint LaBraM, EEGPT, and CBraMod comparisons, we use a common projection head and fixed reconstruction components. Each EEG backbone output is passed through the same lightweight alignment head consisting of a linear layer with GELU activation, a linear projection into the shared alignment space, and layer normalization. The audio encoder remains frozen throughout training. EEG and audio embeddings are $\ell_2$ normalized and optimized with a symmetric CLIP-style cross-entropy loss using a learned logit scale initialized at $\log(1/0.07)$. Unless otherwise stated, pretrained EEG backbones are initialized from the publicly released checkpoints and fine-tuned during alignment. These encoder baselines use the same alignment optimization setup as our method, including AdamW, cosine learning-rate decay with warmup, gradient clipping, mixed precision training, and channel dropout augmentation.

\noindent
\textbf{EEG2Mel.}
We implement EEG2Mel~\citep{eeg2mel} as the direct EEG-to-mel reconstruction baseline, following its original design of predicting mel spectrograms from EEG prior to waveform reconstruction. The input is an EEG PSD tensor of shape $63 \times 125$, and the target is a $64 \times 32$ mel spectrogram. The network consists of five convolutional blocks with output channels $8,16,32,64,128$, kernel size 4, ReLU activations, batch normalization, and dropout in the first four blocks. The resulting representation is max-pooled, flattened, and passed through two 128-dimensional fully connected layers before the final mel-spectrogram regression layer. Training uses mean-squared error on normalized mel spectrograms. For waveform reconstruction, the predicted mel spectrogram is clamped to $[-1,1]$, rescaled to the original decibel range $[-70,14]$, converted to a linear spectrogram via the inverse mel transform with 64 mel bins at a 16 kHz sampling rate, and synthesized using Griffin-Lim. The generated audio is evaluated using the same metrics as all other methods in Table~\ref{tab:res}.

\noindent
\textbf{EEG-conditioned diffusion (Postolache et al.).}
We evaluate the EEG-conditioned latent diffusion method of \citet{postolache2025naturalistic} under the same pooled NMED-T/H setting, one-second inputs, random 95\%/5\% split, and evaluation rules used in this paper. The method conditions a pretrained AudioLDM2 music generator through an EEG-driven ControlNet. Table~\ref{tab:res} reports our evaluation results, including CLAP similarity of 0.672. Audio-level metrics compare generated and target waveforms. For 50-way and 14-way identification, generated-audio CLAP embeddings serve as queries and target-audio CLAP embeddings as candidates, using the candidate-set sampling and duplicate-exclusion rules in Appendix~\ref{app:eval_details}. These identification scores assess audio-output retrieval. The comparison evaluates the complete reconstruction pipelines, while the EEG foundation-model comparisons share our fixed AudioLDM generator.

\noindent
\textbf{LaBraM.}
For LaBraM~\citep{labram}, we use the official NeuralTransformer backbone from the fine-tuning implementation and initialize it with the publicly released checkpoint. The backbone uses a patch size of 200, embedding dimension 200, 12 transformer layers, 10 attention heads, MLP ratio 4, drop-path rate 0.1, and mean pooling. Since NMED EEG windows contain 125 channels, we pass the full 125-channel montage through the LaBraM channel-index interface. Because the alignment windows are shorter than the LaBraM patch size, the time dimension is zero-padded to the nearest patch multiple, scaled by $1/100$, and reshaped into channel-by-patch tensors before being fed into the backbone. We discard the class token, average the remaining patch tokens, and pass the resulting representation to the shared alignment head.

\noindent
\textbf{EEGPT.}
For EEGPT~\citep{wang2024eegpt}, we use the official EEGPTClassifier backbone from the downstream EEGPT implementation and initialize it with the publicly released 58-channel, 4-second EEGPT checkpoint. The model is configured with input size $58\times1024$, patch stride 64, the official 58-channel name list, mean pooling, channel-convolution adaptation enabled, and no predictor or output-projection head. Since the NMED recordings contain 125 channels, the channel-convolution front end is used to map the input into the pretrained EEGPT channel space. Short alignment windows are temporally adapted to the 1024-sample input length expected by the pretrained model. The resulting 512-dimensional EEGPT representation is then passed to the shared alignment head.

\noindent
\textbf{CBraMod.}
For the released-checkpoint CBraMod comparison~\citep{wang2024cbramod}, we use the official implementation and publicly released pretrained weights. The backbone is configured with patch size 200, four temporal patches, model dimension 200, feed-forward dimension 800, 12 transformer layers, and 8 attention heads. Our CBraMod baseline configuration uses 22 input channels. Since NMED contains 125 channels, we learn a $1\times1$ convolutional channel adapter to map the 125-channel input into the pretrained channel space. The time dimension is interpolated to 800 samples when necessary and reshaped into four 200-sample patches. We remove the original output projection, average the backbone features across channel and patch dimensions, and pass the pooled representation to the shared alignment head. Because the CBraMod patch embedding includes an FFT operation that is not reliably supported in bfloat16 on common accelerators, the FFT path is computed in float32 before casting back prior to the alignment head.

\subsection{Hyperparameters} 

Hyperparameters were fixed before evaluation, without validation- or test-based selection. The ridge regularization coefficient is $\lambda=0.1$. Tables \ref{tab:pretrain_hp} and \ref{tab:aligning_hyperparams} describe the reference and eight-second training-only control. The protocol differences are summarized in Table~\ref{tab:pretraining_protocols}. AudioLDM reconstruction uses 100 diffusion sampling steps. 

\begin{table}[ht]
\centering
\caption{Encoder architecture and pretraining configuration. The left block lists the architecture and channel-oriented design parameters. The right block lists the multi-view self-distillation pretraining hyperparameters.}
\label{tab:pretrain_hp}
\small
\begin{tabular}{l c @{\hspace{2em}} l c}
\toprule
\multicolumn{2}{c}{\textbf{Architecture}} & \multicolumn{2}{c}{\textbf{Pretraining}} \\
\cmidrule(lr){1-2}\cmidrule(lr){3-4}
\textbf{Parameter} & \textbf{Value} & \textbf{Parameter} & \textbf{Value} \\
\midrule
Embedding dim       & 512  & Global views          & 2 \\
Patch size          & 50   & Local views           & 8 \\
Num channels        & 125  & Crop (global)         & [0.5, 0.9] \\
SSL drop (student/teacher) & 0.6 / 0 & Crop (local)          & [0.1, 0.5] \\
Aligned dim         & 512  & Noise (global / local) & 0.01 / 0.03 \\
Num heads           & 16   & Learning rate         & 8e-5 \\
Num layers          & 8    & LR scheduler          & Cosine \\
Window size         & 1000 & Warmup steps          & 6000 \\
Window stride       & 800  & Train steps           & 30000 \\
                    &      & Batch size            & 60 \\
\bottomrule
\end{tabular}
\end{table}

\begin{table}[ht]
\centering
\caption{Alignment hyperparameters.}
\label{tab:aligning_hyperparams}
\small
\begin{tabular}{l l c}
\toprule
\textbf{Category} & \textbf{Hyperparameter} & \textbf{Value} \\
\midrule
Data windowing
    & EEG window size      & 125 \\
    & EEG window stride    & 125 \\
\midrule
Optimization
    & Learning rate        & 1.2e-4 \\
    & LR scheduler         & Cosine \\
    & LR warmup steps      & 8000 \\
    & Max train steps      & 30000 \\
    & Max gradient norm    & 1.0 \\
    & Batch size (train / evaluation) & 500 / 50 \\
\midrule
Augmentation
    & Crop scale           & [0.4, 1.0] \\
    & Noise level          & 0.05 \\
    & Channel dropout      & 0.2 \\
\bottomrule
\end{tabular}
\end{table}

\subsection{Hardware and Computation Time} 

All experiments are performed on an NVIDIA RTX 6000 Ada GPU. 
The reference pipeline requires about 20 GPU hours for pretraining and alignment.

\clearpage
\section{Supplementary Experimental Analyses}
\label{app:supplementary_experiments}

\subsection{Training-only and Matched Comparisons}
\label{app:matched}
\textbf{Effect of restricting pretraining to training EEG.}
The eight-second training-only control in Appendix~\ref{app:data_protocols} isolates the effect of test-EEG exposure during SSL while retaining the reference downstream split and other training settings. Table~\ref{tab:train_only} summarizes this control, which yields $0.485\pm0.050$ on 50-way and $0.690\pm0.055$ on 14-way identification (mean $\pm$ SD across three training seeds), with mean CLAP similarity 0.683. Relative to the reference, both identification means decrease by 0.002 and mean CLAP is unchanged, indicating similar average performance when pretraining uses training EEG only.

\begin{table}[ht]
\centering
\caption{Effect of excluding downstream test EEG from SSL pretraining. Reference identification entries reproduce the mean $\pm$ SD over ten evaluation samplings from Table~\ref{tab:res}. Training-only identification entries report mean $\pm$ SD across three training seeds. These SDs describe different sources of variability; CLAP entries report means.}
\label{tab:train_only}
\small
\begin{tabular}{lrrr}
\toprule
Pretraining EEG & 50-way $\uparrow$ & 14-way $\uparrow$ & CLAP $\uparrow$ \\
\midrule
All EEG (reference) & $0.487\pm 0.067$ & $0.692\pm 0.062$ & 0.683 \\
Training-only EEG & $0.485\pm0.050$ & $0.690\pm0.055$ & 0.683 \\
\bottomrule
\end{tabular}
\end{table}

\noindent
\textbf{Matched encoder comparisons.}
The one-second matched comparison in Table~\ref{tab:controls} is a separate experiment from the eight-second training-only control in Table~\ref{tab:train_only}. Our encoder and CBraMod are initialized randomly and pretrained using their respective self-supervised objectives on the same NMED training EEG for 40,000 steps, with the same downstream alignment head, fine-tuning policy, and final-checkpoint rule.

For CSBrain, the released pretrained checkpoint is further pretrained on NMED training EEG. Each one-second NMED window is interpolated from 125 to 200 time samples. A learned $1\times1$ channel projection maps 125 channels to the 19-channel input used by the pretrained model. We use its channel ordering and brain-region partition, followed by the same downstream head and fine-tuning protocol as the training-only comparison.

\subsection{Augmentation Sensitivity Protocol}
\label{app:sensitivity}
The sensitivity results in Section~\ref{sec:ablation} and Table~\ref{tab:sensitivity} vary SSL student channel dropout and temporal crop ranges in separate sweeps. Alignment dropout remains fixed at 0.2 and the teacher has no channel dropout. The dropout sweep compares rates 0.1, 0.2, and 0.3 with the reference SSL student rate of 0.6, using default crops. The crop sweep retains SSL student dropout 0.6. The two default rows in Table~\ref{tab:sensitivity} repeat the same candidate-sampling results for this shared reference configuration. The default global crop range is $[0.5,0.9]$ and the local range is $[0.1,0.5]$. Crop ranges give fractions of the pretraining window. Other settings follow the reference protocol. For each configuration, the trained model is held fixed and identification is evaluated over ten candidate-set samplings. We report the mean and standard deviation across these evaluations.

\subsection{Sensor-level Attention Computation}
\label{app:interpretability}

The alignment encoder retains electrode identity in its channel--time tokens and produces a CLS token alongside them. We extract the attention weights from the CLS query to the channel keys in the final layer, average over temporal-patch keys within each electrode and the 16 attention heads, and then average across time windows within each song. To reduce subject-specific offsets, each subject-level attention map is z-scored across electrodes before averaging across subjects. The resulting per-song attention vector is projected onto the standard 10-10 scalp layout for visualization. The resulting scores quantify attention allocation and retain sensor-level spatial resolution.

\end{document}